\documentclass{article}
\usepackage{graphicx} 
\usepackage[margin=1in]{geometry}
\usepackage{amsmath,amsfonts,amsthm,amssymb}
\usepackage{authblk}
\usepackage{xcolor}
\usepackage{natbib}
\usepackage{url}

\title{Lossy communication constrains iterated learning}
\author[1]{Ben Prystawski}
\author[2]{Dilip Arumugam}
\author[1,2]{Noah D. Goodman}

\date{}

\affil[1]{Department of Psychology, Stanford University}
\affil[2]{Department of Computer Science, Stanford University}

\newif\ifsubmit
\submitfalse
\ifsubmit
\newcommand{\dnote}[1]{}
\newcommand{\bnote}[1]{}
\newcommand{\nnote}[1]{}
\newcommand{\ndg}[1]{}
\else

\newcommand{\dnote}[1]{\textcolor{blue}{Dilip: #1}}
\newcommand{\bnote}[1]{\textcolor{orange}{Ben: #1}}
\newcommand{\nnote}[1]{\textcolor{green}{Noah: #1}}
\newcommand{\ndg}[1]{\nnote{#1}}
\fi

\newcommand{\mc}[1]{\mathcal{#1}}
\newcommand{\ra}{\rightarrow}
\newcommand{\bP}{\mathbb{P}}
\newcommand{\bI}{\mathbb{I}}
\newcommand{\bE}{\mathbb{E}}
\newcommand{\bH}{\mathbb{H}}
\newcommand{\bR}{\mathbb{R}}
\newcommand{\bN}{\mathbb{N}}

\newcommand{\kl}[2]{D_{\mathrm{KL}}(#1 \mid\mid #2)}

\newcommand{\jdiv}[2]{D_{\mathrm{J}}(#1 \mid\mid #2)}
\newcommand{\ubr}[1]{\underbrace{#1}}

\newtheorem{theorem}{Theorem}
\newtheorem{fact}{Fact}

\begin{document}
                            
\maketitle

\begin{abstract}
    Humans' distinctive role in the world can largely be attributed to our capacity for iterated learning, a process by which knowledge is expanded and refined over generations. A range of theories seek to explain why humans are so adept at iterated learning, many positing substantial evolutionary discontinuities in communication or cognition.  Is it necessary to posit large differences in abilities between humans and other species, or could small differences in communication ability produce large differences in what a species can learn over generations? We investigate this question through a formal model based on information theory. We manipulate how much information individual learners can send each other and observe the effect on iterated learning performance. Incremental changes to the channel rate can lead to dramatic, non-linear changes to the eventual performance of the population. We complement this model with a theoretical result that describes how individual lossy communications constrain the global performance of iterated learning. Our results demonstrate that incremental, quantitative changes to communication abilities could be sufficient to explain large differences in what can be learned over many generations.
\end{abstract}
\section{Introduction}

As humans living in the modern world, we have access to an unfathomably vast base of knowledge, from the history of flight to methods for cooking chickpeas. Almost none of this knowledge comes from our individual experience. Instead, human knowledge was built up over many generations through a process of learning and transmitting knowledge. This process of iterated learning has given humans repertoires of knowledge and skills that enable us to invent technologies like snowshoes and air conditioning, allowing us to live in virtually every environment in the world \citep{boyd2011cultural}. This process of incremental improvement over generations is often called the \textit{cultural ratchet} or \textit{cumulative cultural evolution}  \citep{tennie2009ratcheting,mesoudi2018cumulative}. While there is ample historical evidence for human knowledge growing over generations, there is still no agreed-upon answer to the question of why humans exhibit such a strong cultural ratchet compared to other animals.

Common explanations of humans' distinctive capacity for iterated learning focus on communication and social learning. However, the mere presence of social learning cannot be sufficient, as many species learn socially but do not exhibit a ratchet effect \citep{whiten2016cultural}. Chimpanzees learn to crack nuts, macaques learn to wash potatoes, and crows learn to make tools from each other \citep{boesch1994nut,kawai1965newly,whiten1999cultures,kenward2005tool}. Yet none of these species comes close to exhibiting the variety and complexity of culturally-transmitted knowledge we see in humans. Still, humans may exhibit a cultural ratchet because we are \textit{better} social learners than other species. Several explanations for humans' uniqueness appeal to how humans learn from each other, including the ability to pay attention to others' attention \citep{tomasello1999cultural}, improvements in domain-general learning mechanisms and social motivation \citep{heyes2018cognitive}, and specific heuristics governing when to learn socially and who to learn from \citep{boyd2011cultural}.

Some explanations for human cultural distinctiveness appeal to discontinuous cognitive innovations in specific domains, like an innate capacity to process recursive structures \citep{fodor1983modularity,hauser2002faculty,berwick2016only}. Others focus on domain-general improvements in the ability to learn and process information, which are more compatible with incremental changes to individual cognition \citep{heyes2018cognitive}. How big are the differences in communication capacity we need to explain? Do humans need to be \textit{much} better social learners than other animals in order to exhibit a cultural ratchet, or could small differences in the ability to communicate produce large differences in what a species can learn over many generations? If the latter is true, we do not need to appeal to large changes to individual minds to explain the unparalleled breadth and depth of human knowledge. A precise answer to this question might be found through formal modeling of iterated learning.

Many formal models of social cognition study learners in the context of distributed statistical inference. Populations of learners who each receive limited information but can communicate with each other are typically compared to the ideal of Bayesian inference given all the information encountered by all learners. This paradigm was used in seminal work by \citet{beppu2009iterated}, which introduced the idea that language enables ``posterior passing.'' A chain of Bayesian learners who each see limited evidence and pass their posterior distribution to the next chain member, who then uses it as a prior, behaves equivalently to one learner who has access to all of the evidence observed by all generations. Subsequent work has used the paradigm of distributed statistical inference to understand how different social learning heuristics and social interaction mechanisms can enable a population to approximate an ideal learner, and when they fall short \citep{krafft2021bayesian,hardy2023resampling,zubak2024distributed}.
Still, we lack a precise account of the relationship between how good individuals are at communicating and the overall learning of a population.


In this paper, we develop a formal framework to study the relationship between individuals' ability to communicate and what can be learned over generations in a precise, quantitative way. Our framework is rooted in fundamental limits on how much information can be communicated from a sender to a receiver. Our model produces a smooth, yet non-linear, relationship between individual communication capacity and the performance of populations of learners over generations. A fraction of a bit can make the difference between knowledge accumulating and learning hitting a plateau. We complement these simulations with a theoretical result that extends Fano's inequality~\citep{fano1952TransInfoLectNotes}, establishing a relationship between the value of knowledge accumulated at any point in learning and the limits on communication up to that point.

\begin{figure}
    \centering
    \includegraphics[width=\linewidth]{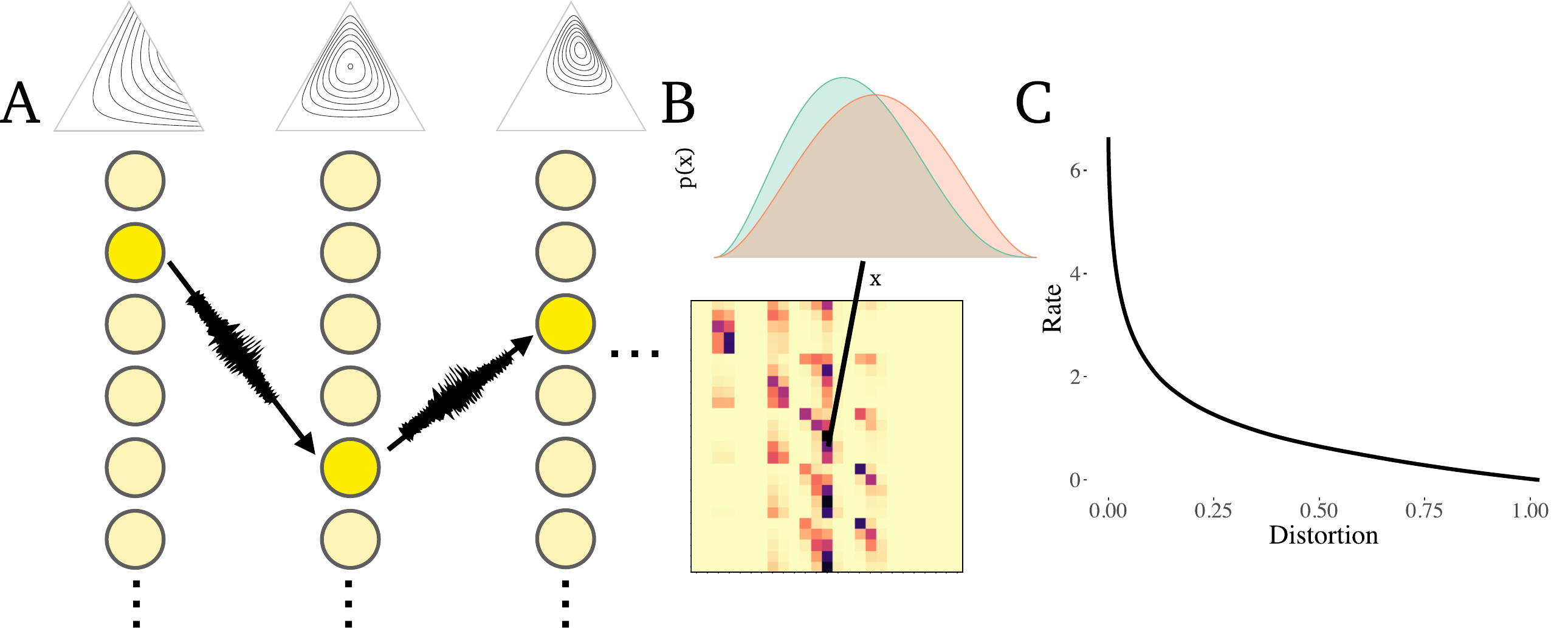}
    \caption{A: Illustration of our iterated learning setup. Learners (yellow circles) observe samples from a true distribution and form posterior beliefs, represented as Dirichlet distributions (top). After observing a fixed number of samples, a learner passes its belief to the next over a rate-limited channel. The next learner starts learning using the received distribution as a prior. B: Illustration of an example channel for a simplified Beta-Bernoulli learning task, where each row corresponds to a transmitted posterior belief and each column represents a received prior. Cells are colored by probability. Example transmitted (green) and received (orange) Beta distributions for a single cell are shown above. C: Illustration of the rate-distortion tradeoff. As the maximum acceptable distortion increases, the necessary rate decreases. Each point along this curve corresponds to a different channel.}
    \label{fig:overview}
\end{figure}

\section{Iterated learning via rate-limited posterior passing}
\label{sec:model}

Our modeling framework builds on the metaphor of iterated learning as posterior passing \citep{beppu2009iterated}.
Our work expands upon existing successful applications of rate-distortion theory to model facets of human communication such as pragmatic language interpretation \citep{zaslavsky2020rate}, language production under cognitive constraints \citep{futrell2021information}, iterated reinforcement learning~\citep{prystawski2023cultural,zhao2025discovering}, as well as human cognition more generally~\citep{sims2016rate,nagy2020optimal,lai2021policy}. The overall modeling setup is illustrated in Figure~\ref{fig:overview}. Code and data are available at \url{https://github.com/benpry/lossy-iterated-learning}

\subsection{Rate-distortion theory}
\label{sec:rdt}

We can represent imperfect communication as a lossy compression problem \citep{shannon1959coding}. In this framing, the communicator's beliefs are the uncompressed data and the recipient's beliefs are the compressed data. 
A lossy compression problem has two components: an information source, which represents uncompressed data, and a distortion function, which quantifies the loss of useful information. For a given information source $p$, one looks to compress data $X \sim p$ via a \emph{channel}. A channel, $p(\cdot \mid x)$, is a conditional probability distribution that assigns probability over compressed outputs $\hat{\mc{X}}$ to each input $x \in \mc{X}$. Given a source $p$, the channel rate $R \in \bR_{\geq 0}$ is the mutual information $\bI(X;\hat{X})$ between the original, uncompressed data $X \sim p$ (channel input) and the compressed data $\hat{X} \sim p(\cdot \mid X)$ (channel output). Intuitively, channel rate is a measure of how much information from the source is successfully transmitted across the channel.
While the natural goal of compression is to discard as many bits of information as possible, some bits are more important than others for any given task. Rate-distortion theory enables us to formalize which bits of information should be prioritized via a distortion function (intuitively, a loss function) $d: \mc{X} \times \hat{\mc{X}} \ra \bR_{\geq 0}$. The distortion function is a non-negative, real-valued function on pairs of uncompressed and compressed data. The value of the distortion function $d(x,\hat{x})$ directly quantifies the loss of fidelity incurred by using $\hat{x} \in \hat{\mc{X}}$ in lieu of $x \in \mc{X}$. For the special case that $\mc{X} = \hat{\mc{X}}$, this is essentially a measure of how bad a given reconstruction is compared to its original source. We are often interested in the expected distortion incurred by a channel $p(\hat{X} \mid X)$ for a given source $p(X)$: $$\bE\left[d(X,\hat{X})\right] = \bE_{X \sim p}\left[\bE_{\hat{X} \sim p(\cdot \mid X)}\left[d(X,\hat{X})\right]\right].$$
We may then ask what the minimum possible expected distortion we can achieve for a given rate is. Rate-distortion theory yields a clear, quantitative answer in the form of the distortion-rate function. For a given rate limit $R \in \bR_{\geq 0}$, the distortion-rate function is defined as $$\mc{D}(R) = \inf\limits_{p(\hat{X} \mid X)} \bE\left[d(X,\hat{X})\right] \text{ such that } \bI(X;\hat{X}) \leq R.$$ 
We can compute channels that achieve the above lower bound on distortion using the classic Blahut-Arimoto algorithm~\citep{blahut1972computation,arimoto1972algorithm,csiszar1974computation}, which is an alternating minimization algorithm that minimizes the (unconstrained) Lagrangian~\citep{boyd2004convex} of the constrained optimization problem: $$\inf\limits_{p(\hat{X} \mid X)} \bI(X, \hat{X}) + \beta \cdot \bE\left[d(X, \hat{X})\right].$$ The hyperparameter $\beta \in \bR_{\geq 0}$ controls the relative importance of minimizing channel rate and distortion. 
We provide more extensive preliminaries, including a detailed derivation of the Blahut-Arimoto algorithm in Appendix~\ref{app:preliminaries}.

\subsection{Limited channels in a Bayesian learning task}

We use a Dirichlet-Categorical learning task to study populations of learners. In the task, there is an unknown probability vector $\mathbf{p}$ over $K \in \bN$ elements. Learners observe samples $X \sim \text{Categorical}(\mathbf{p})$ and update their beliefs about $\mathbf{p}$. A learner's belief is represented as a Dirichlet distribution with parameter vector $\boldsymbol{\alpha}$. 
There is a simple closed-form solution for the parameters of the posterior given the parameters of the prior and the observed sample:
\begin{align*}
    \alpha_j &= \begin{cases}
        \alpha_j + 1 & x_i = j \\
        \alpha_j & \text{otherwise}
    \end{cases}
\end{align*}
Due to the above update rule, Dirichlet parameters are often called ``pseudocounts'' because they keep track of the number of times each value has been observed from the Categorical distribution. 

We measure the accuracy of a belief $\boldsymbol{\alpha}$ using the log probability of the true probability vector $\mathbf{p}$ under the Dirichlet distribution parameterized by $\boldsymbol{\alpha}$. That is, 
\begin{align*}
    \text{score}(\boldsymbol{\alpha}) = \log p (\mathbf{p} ;  \boldsymbol{\alpha})
\end{align*}
Each learner has a limited lifespan $L \in \bN$, meaning it can observe at most $L$ samples from the true distribution. The agent uses this evidence to form a posterior distribution. Learners in the next generation select a member of the previous generation to learn from at random, then receive a compressed version of the teacher's posterior belief. The learner then uses this compressed posterior as a prior and continues learning.

We compute the channel for posterior passing, $p(\hat{\boldsymbol{\alpha}} \mid \boldsymbol{\alpha})$, using the Blahut-Arimoto algorithm. The channel inputs and outputs are both integer-valued parameter vectors. To ensure that the set of possible inputs and outputs is finite, we cap the maximum value of inputs at 21 and outputs at 20. We use a source distribution $p(\boldsymbol{\alpha})$ that is uniform over possible channel inputs and use the Jeffreys divergence~\citep{jeffreys1946invariant}, $D_{\mathrm{J}}$, between the Dirichlet distributions parameterized by $\boldsymbol{\alpha}$ and $\hat{\boldsymbol{\alpha}}$ as a distortion function. Jeffreys divergence is a symmetrized version of Kullback-Leibler (KL) divergence~\citep{kullback1951information}: $\jdiv{p}{q} = \kl{p}{q} + \kl{q}{p}$. Results with different choices of source distribution, distortion function, and lifespan are presented in Appendix~\ref{app:variants}. For a given hyperparameter $\beta$, the Blahut-Arimoto algorithm computes
\begin{align*}
    \inf_{p(\hat{\boldsymbol{\alpha}} | \boldsymbol{\alpha})}\bI(\boldsymbol{\alpha}; \hat{\boldsymbol{\alpha}}) + \beta \cdot \mathbb{E}\Big[ \jdiv{p(\cdot \mid \boldsymbol{\alpha})}{ p(\cdot \mid \hat{\boldsymbol{\alpha}})} \Big]
\end{align*}
By varying the value of $\beta$, we can construct channels that are optimal under different rate limits and study the behavior of populations of agents subject to rate-limited communication over each channel.

This learning task is simple, yet has many desirable properties. Firstly, cumulative learning is possible. It is always possible to increase the log-likelihood of $\mathbf{p}$, and observing more samples from the true distribution reliably does so. The hypothesis space is discrete and small enough to analytically compute optimal channels. While various methods exist to approximate optimal channels, here we can guarantee that results we find are not an artifact of any specific approximation. Finally, learning from evidence is a fundamental task that underpins much of human behavior, like exploration and invention. It is therefore necessary for a population to be able to learn cumulatively from evidence encountered over generations in order to support more complex cultural phenomena.

\section{Performance by generation and channel rate}

We analyze how learners' average score changes over generations for populations with different channel rates.
We can analytically compute the proportion of learners who will have each belief at every timestep, ensuring that our results are not the result of any particular approximation. Details of this analytic formulation are presented in Appendix~\ref{app:analytic}. We use a three-dimensional Dirichlet-Categorical learning task with true $\mathbf{p} = [0.7, 0.2, 0.1]$ and a single observation per generation for our main results.

\subsection{Non-linear effects of channel rate on iterated learning performance}
\begin{figure}
    \centering
    \includegraphics[width=0.49\linewidth]{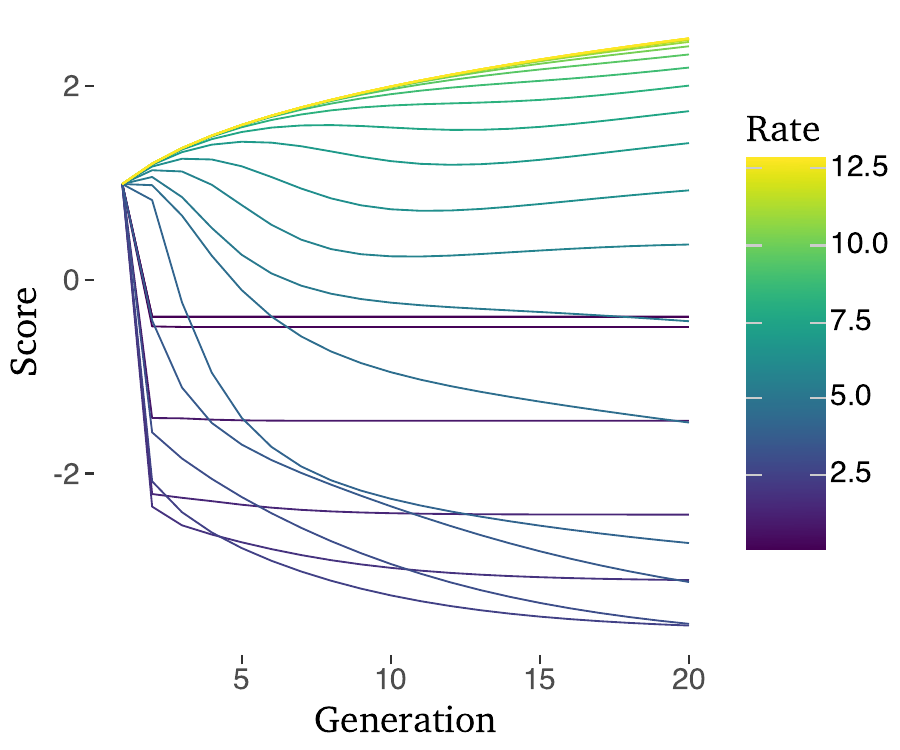}
    \includegraphics[width=0.49\linewidth]{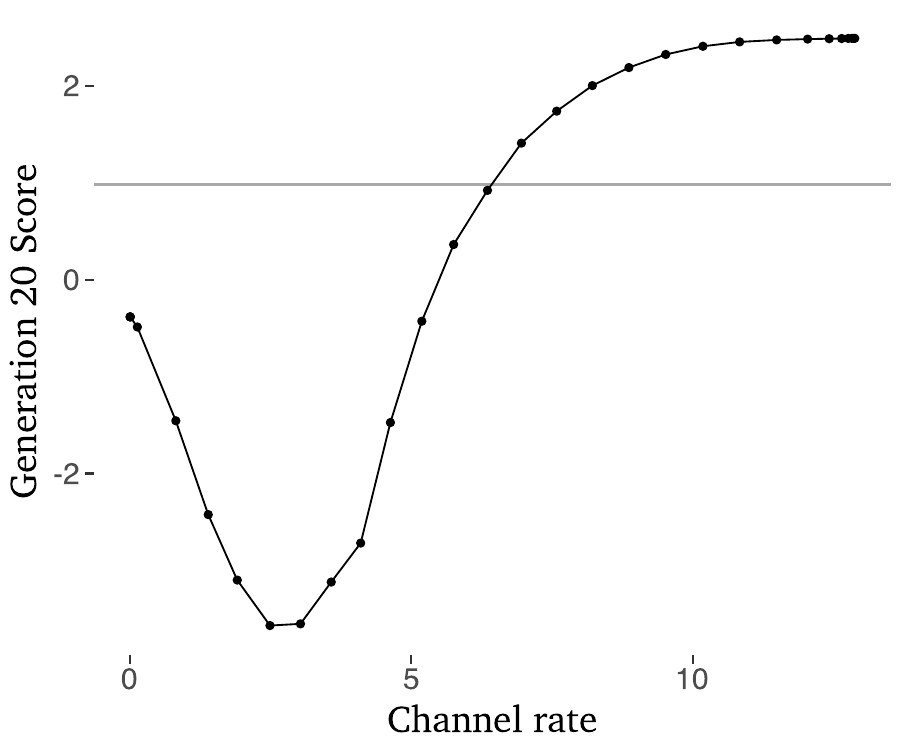}
    \caption{Accuracy of chains of learners' beliefs by generation and channel rate. Left: Scores achieved by learners in each generation of our model by communication rate. Right: Performance at generation 20 by channel rate for communication between individuals. Grey line denotes the expected score of an individual learner starting with a uniform prior and making one observation.}
    \label{fig:performance}
\end{figure}

Figure~\ref{fig:performance} shows the performance of agents in the learning task after 20 generations as a function of the channel rate between individual learners. There is a non-linear relationship between individuals' capacity for communication and iterated learning performance. In the subsequent sections, we analyze why each of these phases occurs and discuss the consequences for theories of cultural learning.

Figure~\ref{fig:channel-viz} visualizes individual channels for a smaller instance of our learning task: a Beta-Bernoulli learning task with a maximum parameter value of 5. In this setting, the set of possible beliefs is small enough that we can visualize exactly how the channel maps each transmitted belief to a received belief.

\subsection{Low rate limits create misinformative communication}

\begin{figure}
    \centering
    \includegraphics[width=\linewidth]{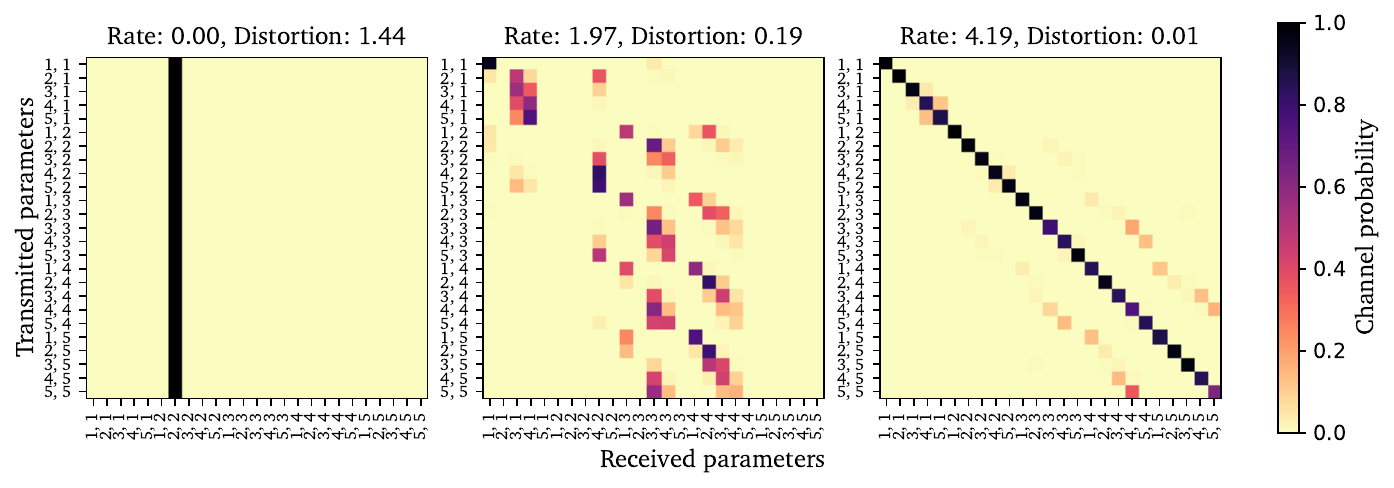}
    \caption{Visualization of channels in a simplified Beta-Bernoulli learning task for three different rate limits. Each row corresponds to a pair of transmitted parameters and  each column corresponds to a pair of received parameters. Cells are colored according to the probability that the channel assigns to the received parameters given the source parameters (yellow=0, purple=1). A channel with a rate of 0 maps every pair of source parameters to the received parameters that minimize expected distortion over all parameters. A channel with a high rate looks diagonal, with every pair of parameters being mapped to itself with high probability.}
    \label{fig:channel-viz}
\end{figure}

When the channel rate is low, subsequent generations perform worse than the first generation. This can be seen in Figure~\ref{fig:performance}, where the score for populations with rate limits below about 7 drops and does not recover.

To understand why this happens, we can consider the case where the rate limit is 0. In this case, the optimal channel maps every possible input to an output distribution that minimizes expected distortion over the entire source distribution. For a uniform source distribution, a channel with minimal rate assigns high probability to $(4, 4, 4)$, representing an informative prior centered around the probabilities $[\frac{1}{3}, \frac{1}{3}, \frac{1}{3}]$. An agent that starts with an informative prior around the wrong probability vector needs to make more observations before the data can overwhelm the prior. In contrast, the first generation starts with a uniform prior, so the posterior it arrives at usually assigns higher probability to $\mathbf{p}$. The first panel of Figure~\ref{fig:channel-viz} shows a channel with a rate limit of 0. This channel maps all inputs to parameters $\alpha=2,\beta=2$ with high probability.

In practice organisms will likely favor individual learning for tasks where they do not have enough communication ability to make learning socially better than nothing. Given that there is a sizable region of rate limits for which received posteriors are worse than uninformative priors, increases in communication rate are not always beneficial. It is only when communication capacity passes a certain threshold (exceeding the horizontal line in Figure~\ref{fig:performance} [right]) that social learning is useful.

\subsection{Intermediate rate limits can create plateaux}

Once learners pass the threshold where iterated learning can out-perform individual learning, iterated learning often exhibits drops in performance for a few generations followed by gradual improvement. Still, learning can reach a plateau after several generations.

Many of the trajectories with rate limits between 7.5 and 10 (blue and dark green lines) in Figure~\ref{fig:performance} (right) show this pattern: Expected performance drops then recovers before eventually leveling off. This result indicates that improvement for a few generations followed by a plateau is a natural consequence of moderately expressive communication.

The middle plot in Figure~\ref{fig:channel-viz} shows the input-output mapping for a channel with an intermediate rate limit. Some pairs of parameters are transmitted accurately, as evidenced by the corresponding cell along the diagonal having high probability. For instance, the channel maps $(1, 4)$ back to itself with high probability. Other beliefs, like $(5, 3)$, are almost never transmitted accurately, as their corresponding diagonal cells have low probability. Some pairs of parameters, like $(5, 5)$, are almost never received regardless of what the original belief is, as their corresponding columns are entirely low-probability (these beliefs are ``hard to talk about'' given the limited communication channel).

Channels where some beliefs are rarely transmitted accurately, like the one shown in Figure~\ref{fig:channel-viz}, can lead to learning plateaux as agents fail to receive important beliefs on the path to cumulative improvement, instead either receiving a less-informative set of parameters (thus negating their individual learning progress) or an inaccurate set of parameters that harms the next generation's ability to learn further.

\subsection{High channel rates produce cumulative learning}

When the channel rate is sufficiently high, learners exhibit cumulative progress, improving their estimates of the true probability incrementally over many generations. The trajectory using the channel with the highest channel rate in Figure~\ref{fig:performance} (left) (yellow line) shows learning with high fidelity (but imperfect) communication. The channel with this rate limit almost always maps parameter vectors back to themselves. Learners with high-fidelity communication improve their estimates of the true probabilities steadily over many generations as information accumulates.

The trajectories with rate limits above 10 (corresponding to the right side of Figure~\ref{fig:performance} [right]) are very close to the trajectory with perfect communication, despite differing by several bits in channel rate. These results indicate that perfect transmission of beliefs is not necessary for cumulative learning over generations. An imperfect, yet high-fidelity, communication channel can also produce cumulative learning. 

\subsection{Discussion}

We have explored three important properties of the relationship between channel rate and iterated learning performance. First, increases in channel rate are not always helpful for iterated learning. Agents in generation 20 would be better off in a chain of learners passing zero bits of information to each other than they would be in a chain that passes two or three bits. In order to truly benefit from communication, learners would need to cross the chasm from 0 bits to the portion of the curve where incremental gains in communication are helpful for iterated learning (around 5 bits). Second, in the region where increases to channel rate are helpful, the curve is quite steep. Small changes in channel rate can compound to create large changes in what a chain of learners can learn over many generations. Finally, there are diminishing returns to channel rate after about 9 bits, where further increases have a negligible effect on performance. This demonstrates that perfect communication is not necessary to learn cumulatively for many generations. Taken together, our simulations show that continuous changes to constraints on individual communication can produce dramatic changes to what can be learned over generations.

Choices of \textit{who} to learn from can also drive the success of iterated learning \citep{thompson2022complex}. In a variant of our model where learners learn selectively from successful members of the previous generation, we find that selective social learning can compensate for lower rate limits to an extent. The details of this analysis are included in Appendix~\ref{app:ssl}.

\section{A variant of Fano's Inequality for iterated learning}

This section presents an extension of Fano's inequality  \citep{fano1952TransInfoLectNotes} that describes how rate limits on individual communication shape the global progress of iterated learning. We prove that the accuracy with which a chain of learners who iteratively pass along beliefs can learn a true distribution is constrained by the amount of information transmitted from one learner to the next. This section describes the framing and the intuition behind the result. The full proof can be found in Appendix~\ref{app:fano-proof}.

As in previous sections, we model the task of iterated learning as inferring a true probability distribution $\mathbf{p}$ from observations. We discretize the space of possible distributions to a set $\{P_v\}_{v \in \mc{V}} \subset \Delta(\mc{X})$ indexed by elements of $\mc{V}$. The task of learning then becomes a hypothesis testing problem~\citep{yu1997assouad,yang1999information}, where the true generating distribution $P_V \in \{P_v\}_{v \in \mc{V}}$ is one of the elements of $\{P_v\}_{v \in \mc{V}}$ and a learner must identify it by observing samples i.i.d. from $P_V$. We assume that the true distribution is drawn uniformly at random from $\{P_v\}_{v \in \mc{V}}$, $V \sim \text{Uniform}(\mc{V})$.

Learners are arranged in a chain. The first learner in the chain begins with a prior belief $\widetilde{B}_1$ over the true distribution. We represent a guess about the true distribution informed by this prior as a random variable $\widehat{V}_1$.  Given a first observation $X_1 \sim P_V$, a learner with prior $\widetilde{B}_1$ arrives at a posterior distribution $B_2$, which is the result of applying Bayes' rule with the prior $\widetilde{B}_1$ and observation $X_1$. The learner then passes their posterior belief $B_2$ through the channel, yielding the compressed belief $\widetilde{B}_2$. As in our simulations, the next learner in the chain receives the compressed belief $\widetilde{B}_2$, which informs a prior guess as to the true distribution, represented by $\widehat{V}_2$. This process may continue across each time period $t \in [T]$, where a learner begins with prior beliefs $\widetilde{B}_t$, observes data from the true distribution $X_t \sim P_V$, updates their posterior beliefs to $B_{t+1}$, and communicates lossily-compressed beliefs $\widetilde{B}_{t+1}$ thereby resulting in an updated guess as to the true distribution $\widehat{V}_{t+1}$.

Per Section \ref{sec:rdt}, we assume that all learner beliefs other than the first generation's prior, $\{\widetilde{B}_{t}\}_{t=2}^T$, are obtained by using the channel that achieves the distortion-rate limit to compress each individual learner's posterior. Ultimately, as the task is to identify the true distribution $P_V$, it is the sequence of random variables representing the learners' refined beliefs, $\{\widehat{V}_t\}_{t=1}^{T+1}$, that determines the overall performance of iterated learning. The following theorem relates the global probability of a chain failing to identify the true distribution and how much local information is transmitted between each individual's final posterior and the next individual's initial prior, across all time periods.

\begin{theorem}
    For iterated learning over $T$ total time periods, we have $$\bP(V \neq \widehat{V}_{T+1}) \geq 1 - \frac{\min\left(\sum\limits_{t=1}^T \bI(\widetilde{B}_t, X_t; \widetilde{B}_{t+1}), \bI(V; \widetilde{B}_1) + \sum\limits_{t=1}^T \bI(V; X_t \mid \widetilde{B}_{t})\right) + \log(2)}{\log(|\mc{V}|)}.$$
    \label{thm:cultural_fano}
\end{theorem}

Theorem \ref{thm:cultural_fano} identifies two dependencies for the quality of iterated learning over a fixed number of time periods $T$: (1) the amount of information communicated between individual learners, as measured by $\bI(\widetilde{B}_t, X_t; \widetilde{B}_{t+1})$, and (2) the task relevance of the communicated information, as measured by the initial informativeness of prior beliefs $\bI(V; \widetilde{B}_1)$ and information gains from observations made in each generation $\bI(V; X_t \mid \widetilde{B}_{t})$. One extreme of the above result occurs when channel rate falls dramatically $R \downarrow 0$. In this case, no amount of time $T$ is able to salvage learning progress as learners' priors $\widetilde{B}_{t+1}$ contain very little information about what the preceding learners encountered. However, this result highlights that equipping learners with high-throughput communication alone is not sufficient to achieve cumulative success in iterated learning. Beyond having an informative prior distribution that helps reduce uncertainty about the reality being learned, a chain of learners also needs informative observations that provide knowledge beyond what is already contained in each current belief $\widetilde{B}_t$. In the extreme case where learners have perfect communication $R \uparrow \infty$, each $\bI(\widetilde{B}_t, X_t; \widetilde{B}_{t+1})$ term is maximized as beliefs are communicated losslessly over the channel. At that point, the informativeness of each observation beyond the current knowledge of efficiently-communicating learners $\bI(V; X_t \mid \widetilde{B}_t)$ governs the success of iterated learning.
To the best of the authors' knowledge, this result is the first to establish a formal relationship between global iterated learning progress and individual communication.

\section{General Discussion}

We have shown that dramatic changes in iterated learning performance can emerge from incremental, quantitative changes in the capacity for communication between individuals. These dynamics appear under quite general assumptions in line with prior models of human learning and communication: Bayesian learning and posterior passing over a rate-limited channel. The exact number of bits that each regime occurs in will depend on the details of the task, but the general trends are robust across different parameter choices. Our results indicate that we may not need to posit large, discontinuous changes to explain substantial differences in what can be learned over many generations. Rather, humans' distinctive capacity for iterated learning could be the result of gradual domain-general improvements to the ability to learn from others.

These results also raise a theoretical problem: given that increases to channel rate are harmful at first, how did social learning evolve to begin with? Increasing the rate from 0 to 2 or 3 leads to worse performance after 20 generations, and social learning performs worse than individual learning with a uniform prior until a minimum channel rate can be achieved (around 6 bits). A solution to this problem might be that some of the necessary channel rate comes from generic learning abilities that evolved for non-social reasons \citep{heyes2015not}. Early social species might have already been in the regime where increases to channel rate are helpful for learning.

Our modeling framework is well-suited to the study of iterated learning for two reasons. It is simple to implement the Bayesian ideal for individual learning, as there is a closed-form solution for a learner's posterior distribution given its prior and observed data. This makes it easy to isolate the effect of communication capacity on iterated learning performance. Furthermore, the space of possible beliefs is discrete and small enough that we can use the Blahut-Arimoto algorithm to compute optimal channels exactly.

However, our task still represents an idealized version of iterated learning, as human reasoning rarely reflects optimal Bayesian inference. Instead, people generally rely on heuristics and approximations \citep{tversky1974judgment,lieder2020resource}. There is also evidence that improvements to information processing capacity within individual minds produce distinctively human capacities for learning and reasoning \citep{cantlon2024uniquely}. Our model may understate the extent to which changes in the individual ability to learn asocially and make new inventions contribute to human culture. Furthermore, incremental improvement of an estimate is a limited form of cumulative learning compared to the historical process of inventing new technologies based on prior innovations. Individual intelligence may be important to explaining the open-endedness of human cultural variation, as measured by the range and variety of culturally-learned skills \citep{morgan2024human}. Future work should study the interplay between communication capacity and individual intelligence in more complex, open-ended domains.

Our model of iterated learning as constrained by a rate limit on individual communication is similar in spirit to resource-rational analysis \citep{lieder2020resource} in both the way it frames the object of study and specific computational methods it uses. Resource-rational analysis frames human cognition as maximizing performance at tasks subject to limits on cognitive resources. Similarly, our model frames iterated learning as maximizing performance on the learning task subject to limited communication between individuals. One framing of resource-rational analysis represents cognitive costs as mutual information between different parts of the mind and uses the Blahut-Arimoto algorithm to compute solutions \citep{arumugam2022rate,arumugam2024bayesian}. Applying resource-rational analysis, possibly in its rate-distortion formulation, to understand other social systems is a promising direction for future work.

The stark differences between humans' capacity for iterated learning and those of other animals can be explained via incremental changes to individuals' ability to communicate and learn from others. Our modeling results tell us that, when seeking to explain humans' unique role in the world, we need not look for a single discontinuity between humans and other animals. Instead, incremental changes in social learning capacity govern whether received information is useless, useful to a limited extent, or increasingly useful generation after generation.

\bibliographystyle{apalike}
\bibliography{citations}

\appendix

\section{Extended preliminaries}
\label{app:preliminaries}

In this section, we provide brief background on the core information-theoretic foundations underlying this work. For any natural number $N \in \bN$, we denote the index set as $[N] \triangleq \{1,2,\ldots,N\}$. For any set $\mc{X}$, $\Delta(\mc{X})$ denotes the set of all probability distributions with support on $\mc{X}$. For any two sets $\mc{X}$ and $\mc{Y}$, we denote the class of all functions mapping from $\mc{X}$ to $\mc{Y}$ as $\{\mc{X} \ra \mc{Y}\} \triangleq \{f \mid f:\mc{X} \ra \mc{Y}\}$. As our exposition throughout the paper will consistently refer to bits of information, all logarithms will be in base $2$.

\subsection{Information theory}
\label{sec:info_theory}

Here we introduce core concepts in probability theory and information theory~\citep{shannon1948mathematical} used throughout this paper. See \citet{cover2012elements,gray2011entropy,polyanskiy2022IT} and \citet{duchi23ItLectNotes} for more background. 

We define the mutual information between any two random variables $X,Y$ through the Kullback-Leibler (KL) divergence~\citep{kullback1951information}: $$\bI(X;Y) = \kl{p(X, Y)}{p(X) \times p(Y)} \qquad \kl{p}{q} = \begin{cases} \bE_p\left[\log\left(\frac{p(X)}{q(X)}\right)\right] & p \ll q \\ +\infty & p \not\ll q \end{cases},$$ where $p$ and $q$ are both probability distributions on the same set $\mc{X}$. An analogous definition of conditional mutual information holds through the expected KL-divergence for any three random variables $X,Y,Z$:
$$\bI(X;Y \mid Z) = \bE\left[\kl{p(X, Y \mid Z)}{p(X \mid Z) \times p(Y \mid Z)}\right].$$
With these definitions in hand, we can define the entropy and conditional entropy for any two random variables $X,Y$ as $$\bH(X) = \bI(X;X) \qquad \bH(Y \mid X) = \bH(Y) - \bI(X;Y).$$ This yields the following identities for mutual information and conditional mutual information for any three arbitrary random variables $X$, $Y$, and $Z$:
$$\bI(X;Y) = \bH(X) - \bH(X \mid Y) = \bH(Y) - \bH(Y | X), \qquad \bI(X;Y|Z) = \bH(X|Z) - \bH(X \mid Y,Z) = \bH(Y|Z) - \bH(Y | X,Z).$$
Through the chain rule of the KL-divergence and the fact that $\kl{p}{p} = 0$ for any probability distribution $p$, we obtain another equivalent definition of mutual information, $$\bI(X;Y) = \bE\left[\kl{p(Y \mid X)}{p(Y)}\right],$$ as well as the chain rule of mutual information: $\bI(X;Y_1,\ldots,Y_n) = \sum\limits_{i=1}^n \bI(X;Y_i \mid Y_1,\ldots,Y_{i-1}).$ The so-called ``golden formula'' of mutual information (see Theorem 4.1 and Corollary 4.2 of \citep{polyanskiy2022IT}) establishes that, for any distribution $q(Y) \in \Delta(\mc{Y})$, we have $$\bI(X; Y) = \bE\left[\kl{p(Y \mid X)}{q(Y)}\right] - \kl{p(Y)}{q(Y)},$$ which implies that $$\bI(X;Y) = \inf\limits_{q(Y)} \bE\left[\kl{p(Y \mid X)}{q(Y)}\right],$$ when $\bI(X;Y) < \infty$.

Although the KL-divergence is widely used, it is asymmetric. In cases where there is no obvious natural direction in which to compute the KL-divergence, we can use the Jeffreys divergence (see $I_2$ in Equation 1 of \citet{jeffreys1946invariant} or, perhaps more transparently, Equation 3.2 of \citet{kullback1959information}) $D_J$, which is a symmetrized version of KL divergence. $$\jdiv{p}{q} = \kl{p}{q} + \kl{q}{p}$$ where $p$ and $q$ are two arbitrary probability distributions with support on $\mc{X}$. 

We conclude this section with a series of classic results from information theory which we use to establish our main theoretical result. 

\begin{fact}[Data-Processing Inequality~\citep{cover2012elements}]
For any Markov chain $X - Y - Z$, $\bI(X;Z) \leq \bI(X; Y).$
\label{fact:dpi}
\end{fact}

\begin{fact}[Data-Processing Inequality on Four Random Variables]
For any Markov chain $U - X - Y -Z$, $\bI(U; Z) \leq \bI(X;Y).$
\label{fact:dpi_4rv}
\end{fact}
\begin{proof}
    We begin by examining the mutual information $\bI(U, X; Y, Z)$ and applying the chain rule of mutual information twice to obtain $$\bI(U, X; Y, Z) = \bI(X; Y, Z) + \bI(U; Y, Z \mid X) = \bI(X; Y) + \bI(X; Z \mid Y) + \bI(U; Y, Z \mid X).$$ Observe that the Markov chain $U - X - Y - Z$ implies the Markov chain $X - Y - Z$. By definition, a Markov chain implies that the ``past'' and ``future'' are conditionally independent given the ``present.'' Consequently, the Markov chain $X - Y - Z$ implies $\bI(X; Z \mid Y) = 0$. Similarly, define the random variable $W = (Y, Z)$ and observe that the Markov chain $U - X - Y - Z$ implies the Markov chain $U - X - W$, such that $\bI(U; Y, Z \mid X) = \bI(U; W \mid X) = 0$. Thus, we have that $$\bI(U, X; Y, Z) = \bI(X; Y) + \ubr{\bI(X; Z \mid Y)}_{= 0} + \ubr{\bI(U; Y, Z \mid X)}_{= 0} = \bI(X; Y).$$ Meanwhile, expanding the original mutual information term via the chain rule of information again yields $$\bI(U, X; Y, Z) = \bI(U; Y, Z) + \bI(X; Y, Z \mid U) = \bI(U; Z) + \ubr{\bI(U; Y \mid Z)}_{\geq 0} + \ubr{\bI(X; Y, Z \mid U)}_{\geq 0} \geq \bI(U; Z),$$ where the inequality follows by the non-negativity of (conditional) mutual information. Putting everything together, we have established that $$\bI(U; Z) \leq \bI(U, X; Y, Z) = \bI(X;Y),$$ as desired.
\end{proof}

\begin{fact}[Fano's Inequality~\citep{fano1952TransInfoLectNotes}]
Let $V$ be a discrete random variable taking values on a finite set $\mc{V}$ and consider any Markov chain $V - X - \widehat{V}$. If $V$ is uniformly distributed ($V \sim \text{Uniform}(\mc{V})$), then
\begin{align*}
    \bP(V \neq \widehat{V}) &\geq 1 - \frac{\bI(V; \widehat{V}) + \log(2)}{\log(|\mc{V}|)} \\
    &\geq 1 - \frac{\bI(V; X) + \log(2)}{\log(|\mc{V}|)},
\end{align*}
where the final inequality follows via the data-processing inequality (Fact \ref{fact:dpi}). Moreover, if $P_{\mathrm{max}} \triangleq \max\limits_{v \in \mc{V}} \bP(V = v)$, then $$\bP(V \neq \widehat{V}) \leq 1 - \frac{\bI(V; X) + \log(2)}{\log(\frac{1}{P_{\mathrm{max}}})}.$$
\label{fact:fano}
\end{fact}
Proofs for both statements of Fano's inequality can be found as Equations 3.18 and 3.19 of \citet{polyanskiy2022IT}. A more accessible proof of the first inequality (which is more commonly recognized as Fano's inequality) that proceeds via entropy definitions instead of the data-processing inequality can be found as Proposition 2.3.3 of \citep{duchi23ItLectNotes}.

Finally, we note that there is a natural corollary of Fano's inequality that conditions on an arbitrary third random variable. For ease of exposition, we will refer to this conditional variant as Fano's inequality as well, which can be proven by applying the traditional Fano's inequality pointwise --- that is, for $\bP(V \neq \widehat{V} \mid Z =z)$ and $\bI(V; \widehat{V} \mid Z = z)$ to any $z \in \mc{Z}$ --- and then taking expectations on both sides.

\begin{fact}[Conditional Fano's Inequality]
    Let $V$ be a discrete random variable taking values on a finite set $\mc{V}$ and let $Z$ be an arbitrary random variable such that, conditioned on $Z$, we have a Markov chain $V - X - \widehat{V}$. If $V$ is uniformly distributed ($V \sim \text{Uniform}(\mc{V})$), then
\begin{align*}
    \bP(V \neq \widehat{V}) &\geq 1 - \frac{\bI(V; \widehat{V} \mid Z) + \log(2)}{\log(|\mc{V}|)} \\
    &\geq 1 - \frac{\bI(V; X \mid Z) + \log(2)}{\log(|\mc{V}|)}.
\end{align*}
\label{fact:cond_fano}
\end{fact}

\subsection{The Blahut-Arimoto algorithm}
\label{sec:blahut-arimoto}

Conveniently, the optimization problem in the distortion-rate function is convex~\citep{chiang2004geometric}. The seminal Blahut-Arimoto algorithm solves this problem using the (unconstrained) Lagrangian of the constrained optimization problem~\citep{blahut1972computation,arimoto1972algorithm,csiszar1974computation,boyd2004convex}:
$$
\inf\limits_{p(\hat{X} \mid X)} \bI(X, \hat{X}) + \beta \cdot \bE\left[d(X, \hat{X})\right].
$$
For ease of exposition, we focus our discussion in this section on the computation of the rate-distortion function, which is the dual of the distortion-rate function:
$$
\mc{R}(D) = \inf\limits_{p(\hat{X} \mid X)} \bI(X;\hat{X}) \text{ such that }  \bE\left[d(X,\hat{X})\right] \leq D.
$$ 
The derivation is near-identical for the distortion-rate function and the Blahut-Arimoto algorithm ultimately comes out the same.  Thus, rather than dealing with an explicit rate limit $R$ or distortion threshold $D$, the Blahut-Arimoto algorithm is an alternating-maximization algorithm governed by a hyperparameter $\beta \in \bR_{\geq 0}$, with lower $\beta$ values representing a preference for minimizing rate and higher $\beta$ values representing a preference for minimizing expected distortion. Geometrically, this $\beta$ corresponds to the (negated) slope of a tangent line to the rate-distortion function, which is known to be non-negative, convex, and non-increasing in $D$. Thus, to each specific value of $\beta$, the Blahut-Arimoto algorithm computes the channel $p(\widehat{X} \mid X)$ associated with the point on the rate-distortion curve that coincides with tangent line of slope $-\beta$. Relating a particular distortion threshold $D \in \bR_{\geq 0}$ to a concrete value of $\beta$ follows from the fact that strong duality holds~\citep{csiszar1974extremum}: $$\mc{R}(D) = \sup\limits_{\beta \in \bR_{\geq 0}} \left[\left(\inf\limits_{p(\hat{X} \mid X)} \bI(X, \hat{X}) + \beta \cdot \bE\left[d(X, \hat{X})\right]\right) - \beta \cdot D\right].$$ From a computational perspective, the Blahut-Arimoto algorithm operates under the assumption that all random variables involved in the lossy compression problem are discrete, however the algorithm itself does apply for abstract random variables~\citep{csiszar1974computation,csiszar1974extremum}.

Beginning with the Lagrangian induced from the original rate-distortion optimization, $$\inf\limits_{p(\hat{X} \mid X)} \bI(X, \hat{X}) + \beta \cdot \bE\left[d(X, \hat{X})\right],$$ the Blahut-Arimoto algorithm first proceeds by introducing a variational approximation $q(\widehat{X}) \in \Delta(\widehat{\mc{X}})$ to the true marginal distribution $p(\widehat{X} = \widehat{x}) = \bE\left[p(\widehat{X} = \widehat{x} \mid X)\right]$. Due to the golden formula of mutual information (see Section \ref{sec:info_theory}), the Lagrangian can be rewritten as a bi-level optimization problem $$\inf\limits_{p(\hat{X} \mid X)} \bI(X, \hat{X}) + \beta \cdot \bE\left[d(X, \hat{X})\right] = \inf\limits_{p(\hat{X} \mid X)} \inf\limits_{q(\widehat{X})} \bE_{X \sim p}\left[\kl{p(\widehat{X} \mid X)}{q(\widehat{X})}\right] + \beta \cdot \bE\left[d(X, \hat{X})\right],$$ with separate updates for the channel $p(\widehat{X} \mid X)$ and channel marginal approximation $q(\widehat{X})$. For any fixed channel $p(\widehat{X} \mid X)$, we have that $$\inf\limits_{q(\widehat{X})} \bE_{X \sim p}\left[\kl{p(\widehat{X} \mid X)}{q(\widehat{X})}\right] = \inf\limits_{q(\widehat{X})} \bI(X; \widehat{X}) + \kl{p(\widehat{X})}{q(\widehat{X})} = \inf\limits_{q(\widehat{X})} \kl{p(\widehat{X})}{q(\widehat{X})} = 0,$$ where $p(\widehat{X})$ is the marginal induced by the fixed channel $p(\widehat{X} \mid X)$ and the infimum is achieved when $q(\widehat{X} = \widehat{x}) = p(\widehat{X} = \widehat{x}) = \bE\left[p(\widehat{X} = \widehat{x} \mid X)\right]$ for all $\widehat{x} \in \widehat{\mc{X}}$.

For a fixed channel marginal $q(\widehat{X})$, the update for the channel requires a slightly more involved derivation. With the channel marginal fixed, we aim to solve $$\inf\limits_{p(\hat{X} \mid X)} \bE_{X \sim p}\left[\kl{p(\widehat{X} \mid X)}{q(\widehat{X})}\right] + \beta \cdot \bE\left[d(X, \hat{X})\right] = -\sup\limits_{p(\hat{X} \mid X)} \bE\left[-\beta d(X, \hat{X})\right] - \bE_{X \sim p}\left[\kl{p(\widehat{X} \mid X)}{q(\widehat{X})}\right],$$ where the equation holds since we've only multiplied by $1 = (-1) \cdot (-1)$ and brought one negative into the infimum, changing it to a supremum. Thus, the channel update amounts to finding the distribution that achieves this supremum. Next, we leverage a variational form of the KL-divergence known as the Donsker-Varadhan representation~\citep{donsker1983asymptotic} which, for two distributions $p,q \in \Delta(\mc{X})$, establishes that $$\kl{p}{q} = \sup\limits_{g \in \{\mc{X} \ra \bR\}} \bE_{X \sim p}\left[g(X)\right] - \log\left(\bE_{X \sim q}\left[\exp\left(g(X)\right)\right]\right).$$ As the supremum can be lower bounded by any choice of function $g \in \{\mc{X} \ra \bR\}$, we can apply this lower bound to the KL-divergence term and re-arrange terms to rewrite the previous equation as:
\begin{center}
\resizebox{\textwidth}{!}{%
$\begin{aligned}
    \sup\limits_{p(\hat{X} \mid X)} \bE\left[-\beta d(X, \hat{X})\right] - \bE_{X \sim p}\left[\kl{p(\widehat{X} \mid X)}{q(\widehat{X})}\right] &=\sup\limits_{p(\hat{X} \mid X)} \bE_{X \sim p}\left[\bE_{\widehat{X} \sim p(\cdot \mid X)}\left[-\beta d(X, \hat{X})\right] - \kl{p(\widehat{X} \mid X)}{q(\widehat{X})}\right] \\
    &\leq \bE_{X \sim p}\left[\log\left(\bE_{\widehat{X} \sim q}\left[\exp\left(-\beta d(X,\widehat{X})\right)\right]\right)\right].
\end{aligned}$}
\end{center}

Since the supremum is over all possible channels, consider the channel $$p(\widehat{X} \mid X) = \frac{q(\widehat{X})\exp\left(-\beta d(X,\widehat{X})\right)}{\bE_{\widehat{X} \sim q}\left[\exp\left(-\beta d(X,\widehat{X})\right)\right]}$$ and observe that 
\begin{align*}
    \bE_{X \sim p}&\left[\kl{p(\widehat{X} \mid X)}{q(\widehat{X})}\right] = \bE_{X \sim p}\left[\bE_{\widehat{X} \sim p(\cdot \mid X)}\left[\log\left(\frac{p(\widehat{X} \mid X)}{q(\widehat{X})}\right)\right]\right] \\
    &= \bE_{X \sim p}\left[\bE_{\widehat{X} \sim p(\cdot \mid X)}\left[\log\left(p(\widehat{X} \mid X)\right)\right] - \bE_{\widehat{X} \sim p(\cdot \mid X)}\left[\log\left(q(\widehat{X})\right)\right]\right] \\
    &= \bE_{X \sim p}\left[\bE_{\widehat{X} \sim p(\cdot \mid X)}\left[\log\left(q(\widehat{X})\right)\right] + \bE_{\widehat{X} \sim p(\cdot \mid X)}\left[\log\left(\frac{\exp\left(-\beta d(X,\widehat{X})\right)}{\bE_{\widehat{X} \sim q}\left[\exp\left(-\beta d(X,\widehat{X})\right)\right]}\right)\right]  - \bE_{\widehat{X} \sim p(\cdot \mid X)}\left[\log\left(q(\widehat{X})\right)\right]\right] \\
    &= \bE_{X \sim p}\left[\bE_{\widehat{X} \sim p(\cdot \mid X)}\left[\log\left(\frac{\exp\left(-\beta d(X,\widehat{X})\right)}{\bE_{\widehat{X} \sim q}\left[\exp\left(-\beta d(X,\widehat{X})\right)\right]}\right)\right] \right] \\
    &= \bE_{X \sim p}\left[\bE_{\widehat{X} \sim p(\cdot \mid X)}\left[-\beta d(X, \hat{X})\right] - \log\left(\bE_{\widehat{X} \sim q}\left[\exp\left(-\beta d(X,\widehat{X})\right)\right]\right)\right].
\end{align*}

Substituting back into the earlier expression, we obtain
\begin{align*}
    \bE_{X \sim p}\left[\bE_{\widehat{X} \sim p(\cdot \mid X)}\left[-\beta d(X, \hat{X})\right] - \kl{p(\widehat{X} \mid X)}{q(\widehat{X})}\right] &=\bE_{X \sim p}\left[\log\left(\bE_{\widehat{X} \sim q}\left[\exp\left(-\beta d(X,\widehat{X})\right)\right]\right)\right].
\end{align*}

Thus, we've shown $$\bE_{X \sim p}\left[\log\left(\bE_{\widehat{X} \sim q}\left[\exp\left(-\beta d(X,\widehat{X})\right)\right]\right)\right] = \sup\limits_{p(\hat{X} \mid X)} \bE_{X \sim p}\left[\bE_{\widehat{X} \sim p(\cdot \mid X)}\left[-\beta d(X, \hat{X})\right] - \kl{p(\widehat{X} \mid X)}{q(\widehat{X})}\right],$$ (also known as the Gibbs variational principle - see Proposition 4.7 of \citep{polyanskiy2022IT}) and that the channel $$p(\widehat{X} \mid X) = \frac{q(\widehat{X})\exp\left(-\beta d(X,\widehat{X})\right)}{\bE_{\widehat{X} \sim q}\left[\exp\left(-\beta d(X,\widehat{X})\right)\right]}$$ achieves the supremum, thereby also achieving the infimum in the original Lagrangian of the rate-distortion function. Consequently, the Blahut-Arimoto algorithm emerges as iterating the pair of alternating update equations between the channel and channel marginal, respectively: $$q(\widehat{X}) = \bE\left[p(\widehat{X} \mid X)\right], \qquad p(\widehat{X} \mid X) = \frac{q(\widehat{X})\exp\left(-\beta d(X,\widehat{X})\right)}{\bE_{\widehat{X} \sim q}\left[\exp\left(-\beta d(X,\widehat{X})\right)\right]}.$$

\section{Results with different model variants}
\label{app:variants}

To test how robust our findings are to different formulations of our model, we consider three variants of the model: a variant with a non-uniform source distribution, a variant where each individual makes more observations, and a variant that uses a different true probability vector.

\subsection{Non-uniform source distribution}

First, we consider a model where the source probability decreases exponentially in the sum of the pseudocounts. This variant favors parameters that correspond to relatively flat distributions with high entropy, as it assigns higher probability to distributions with low pseudocounts. The source probability is proportional to the following:
\begin{align*}
    p(\boldsymbol{\alpha}) &\propto 1.1^{-(\alpha_1 + \alpha_2 + \alpha_3)}
\end{align*}

This source distribution has the additional advantage of ensuring that the source distribution is well-defined over all possible Dirichlet parameter vectors. With exponentially decreasing probability, we could construct a distribution that takes support over all possible pseudocounts, not just those less than or equal to 20.

Figure~\ref{fig:decreasing-variant} shows the results of simulations with the decreasing-probability model variant. These results are virtually identical to the results with a uniform source distribution, suggesting that the behavior of our simulated populations is robust over reasonable choices of source distribution.

\begin{figure}
    \centering
    \includegraphics[width=0.48\linewidth]{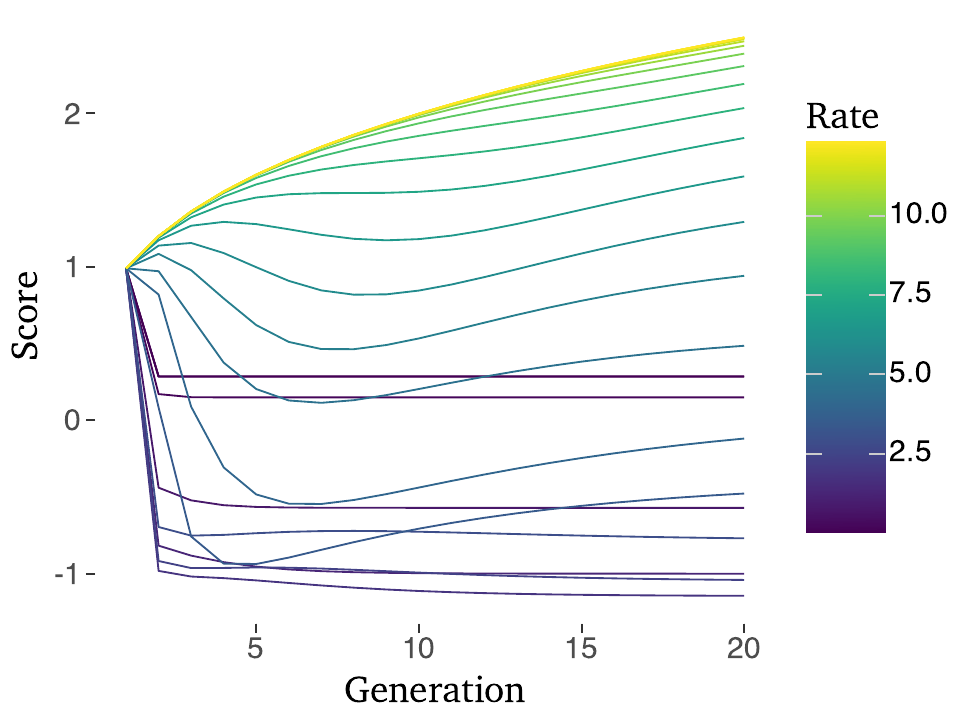}
    \includegraphics[width=0.48\linewidth]{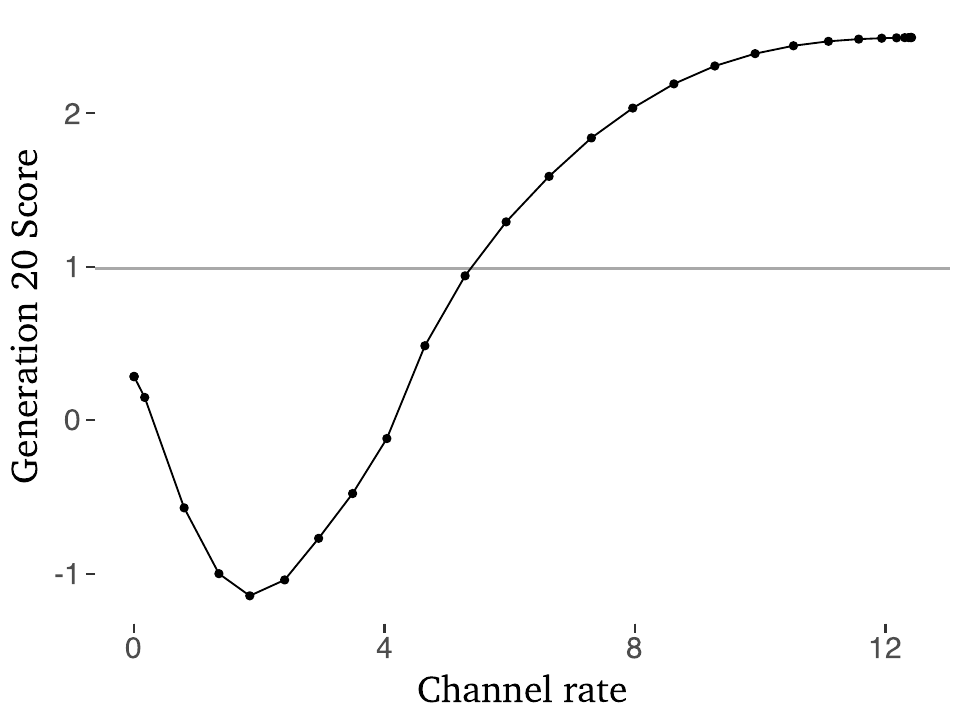}
    \caption{Performance by generation and performance at last generation by rate for the model variant with exponentially decreasing source probability}
    \label{fig:decreasing-variant}
\end{figure}

\subsection{Longer individual lifespan}

In this model variant, we increase each learner's effective individual lifespan by giving them two samples from the true distribution instead of one. We run chains of length 10 rather than 20, to prevent the agents from exceeding the maximum parameter value. Figure~\ref{fig:longlife-variant} shows results for this variant. Again, the results with this variant are very close to those of the original model, though learning happens faster as each agent receives more evidence.

\begin{figure}
    \centering
    \includegraphics[width=0.48\linewidth]{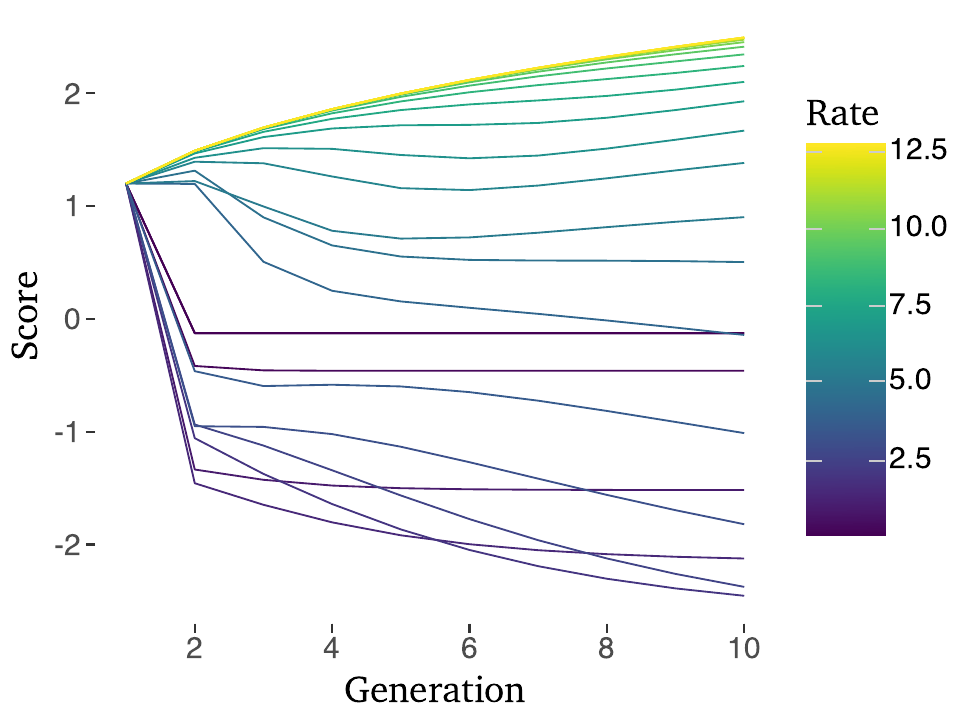}
    \includegraphics[width=0.48\linewidth]{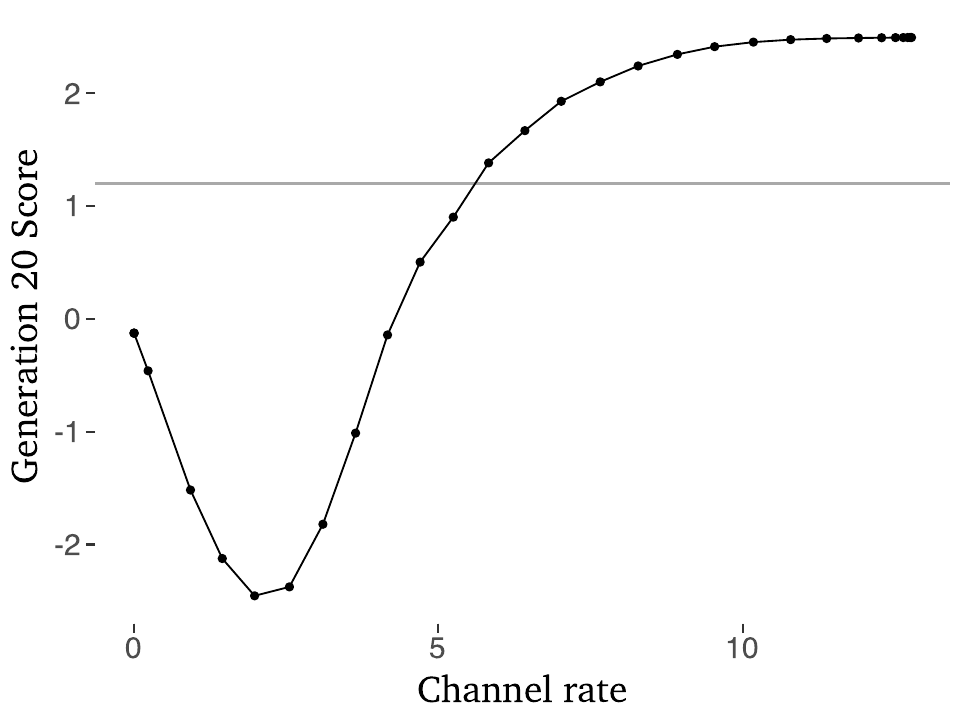}
    \caption{Performance by generation and performance at last generation by rate for the model variant with two observations per agent rather than one}
    \label{fig:longlife-variant}
\end{figure}

\subsection{Different true probability vectors}

\begin{figure}
    \centering
    \includegraphics[width=0.48\linewidth]{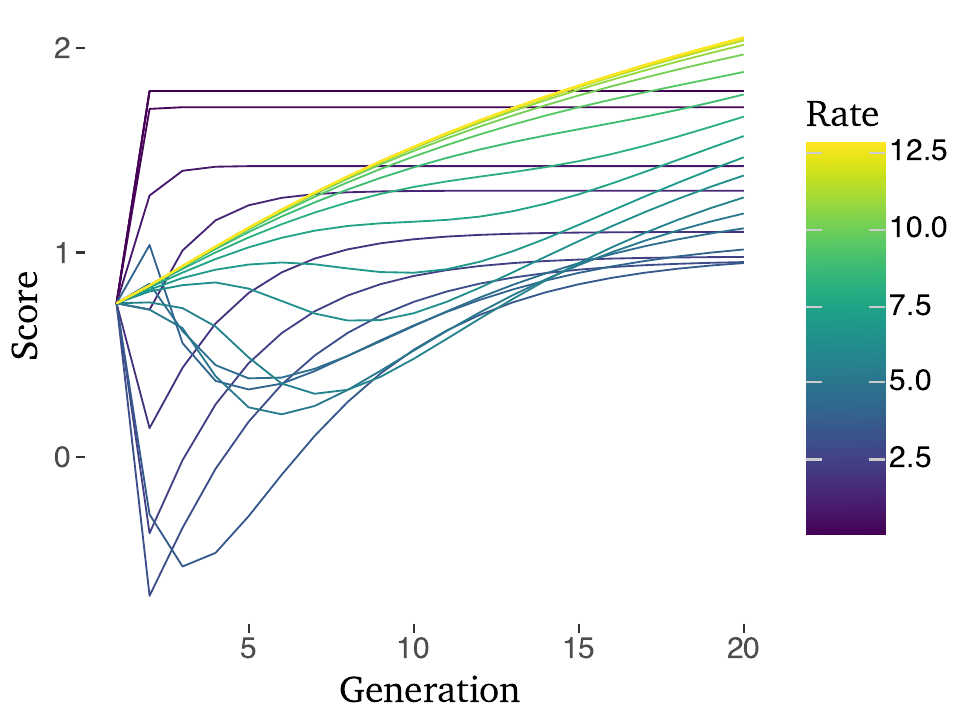}
    \includegraphics[width=0.48\linewidth]{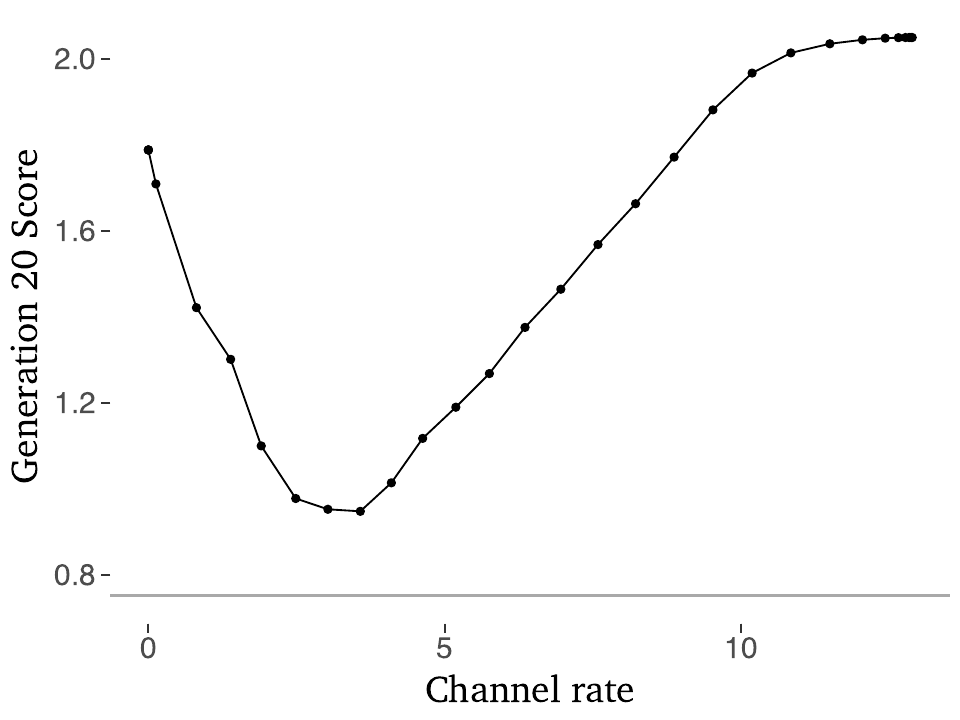}
    \caption{Performance by generation and performance at last generation by rate for the model variant with true probabilities closer to uniform: $[0.5, 0.25, 0.25]$}
    \label{fig:closertrueprobs-variant}
\end{figure}

\begin{figure}
    \centering
    \includegraphics[width=0.48\linewidth]{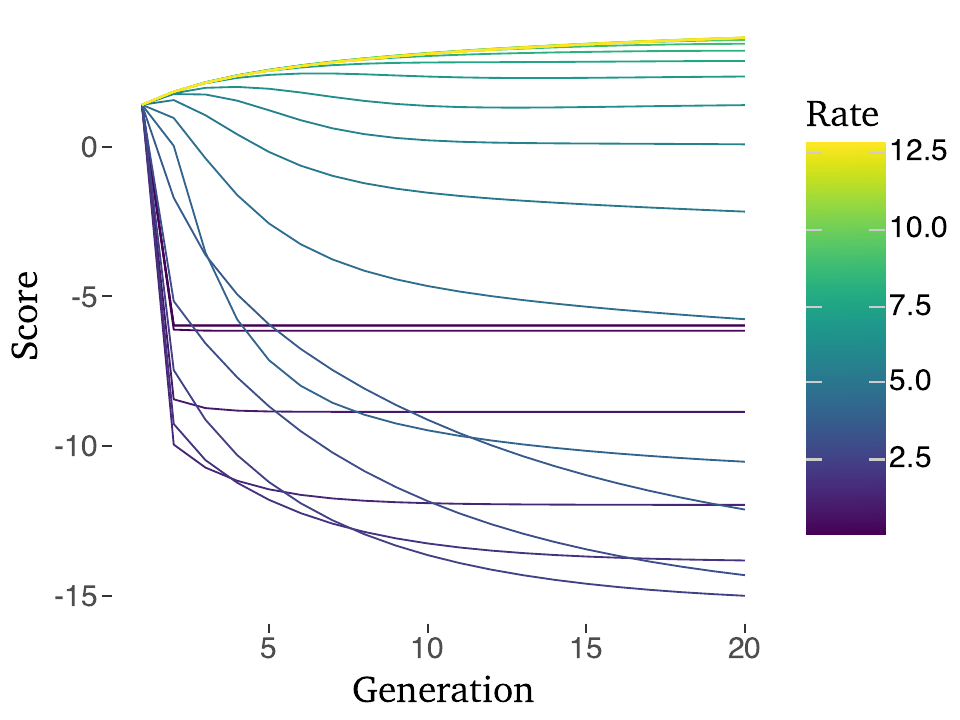}
    \includegraphics[width=0.48\linewidth]{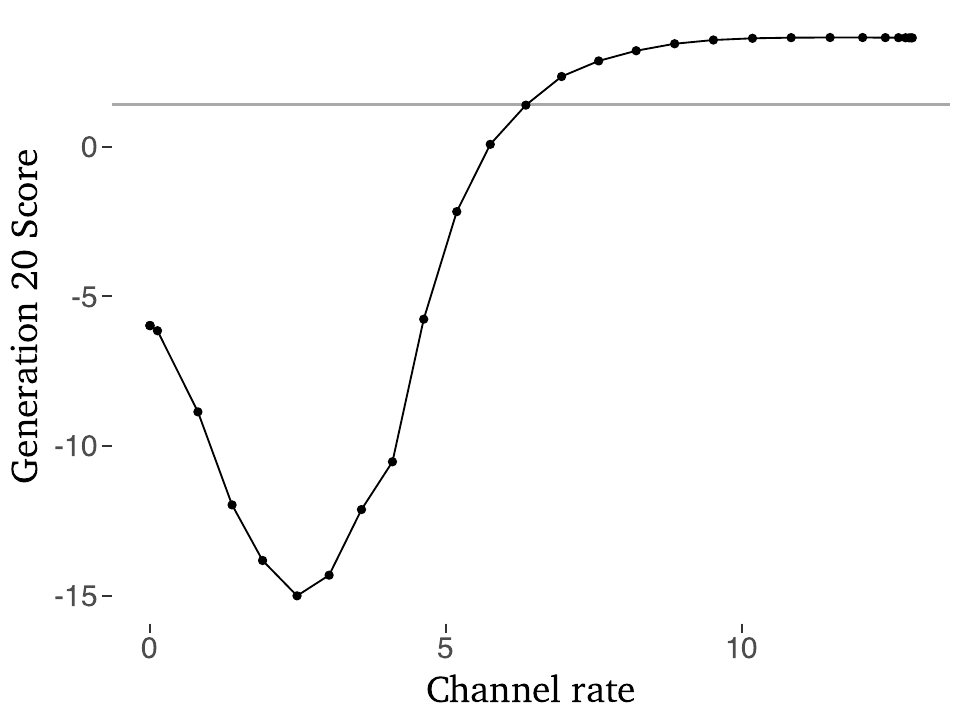}
    \caption{Performance by generation and performance at last generation by rate for the model variant with true probabilities farther from uniform: $[0.9, 0.05, 0.05]$}
    \label{fig:farthertrueprobs-variant}
\end{figure}

We test the trajectory of iterated learning for two different sets of true probabilities. We choose $\mathbf{p} = [0.5, 0.25, 0.25]$ as an closer to the center of the Dirichlet simplex and $\mathbf{p} = [0.9, 0.05, 0.05]$ as an example further away. 

For $\mathbf{p} = [0.5, 0.25, 0.25]$, the overall pattern of performance by generation and rate looks somewhat different to the pattern in the main results. Chains of learners out-perform individuals in generation 1 even with very low rate limits. These probabilities are closer to the uniform distribution $p = [\frac{1}{3}, \frac{1}{3}, \frac{1}{3}]$. Channels with low rates consistently map parameters to distributions that are peaked around the center, so lossy communication happens to be helpful. Of course, this only works when the true state of the world happens to be close to what the channel distorts beliefs toward.

For $\mathbf{p} = [0.9, 0.05, 0.05]$, the pattern of performance looks like a more dramatic version of the main results. The true probabilities are even further from uniform than in the main simulations, so the pattern of results is more dramatic. Lossy communication is even worse: a rate limit near 0 leads to a score of -5 compared to a score just under 0 in the main simulations.
\section{Analytic representation of simulations}
\label{app:analytic}

The model described in Section~\ref{sec:model} describes individual agents learning and passing beliefs from one to the next. While we can simulate individual agents, we can also analytically compute the proportion of agents who will hold each belief at each time point. This section describes how.

First, we define a belief prevalence distribution as a vector $\mathbf{b}_t$ as a distribution over possible beliefs that an agent at time $t$ could have. Since the space of beliefs is finite and discrete, we can represent $\mathbf{b}_t$ as a vector in $\mathbb{R}^n$, where $n$ is the number of distinct beliefs (parameter combinations) in the space. Each element of the vector corresponds to a different belief.

We can then represent the updates to beliefs that result from observation via an \textit{observation transition matrix} $O \in \mathbb{R}^{n \times n}$. We compute the beliefs for timestep $t+1$ by multiplying the beliefs from the prior timestep by $O$. That is, $$\mathbf{b}_{t+1} = O\mathbf{b}_t$$ 

When agents transmit beliefs from one generation to the next, we multiply the posterior belief prevalence vector $\mathbf{b}_{t+1}$ by the channel matrix $C$. Therefore, the belief prevalence distribution after $g$ generations of iterated learning is
\begin{align*}
    \mathbf{b}_g &= O(CO)^{g-1}\mathbf{b}_0
\end{align*}

We can simply perform these multiplications, then average the score of each belief after each generation to compute the expected score in each generation.

\section{Selective social learning}
\label{app:ssl}

\begin{figure}
    \centering
    \includegraphics[width=0.61\linewidth]{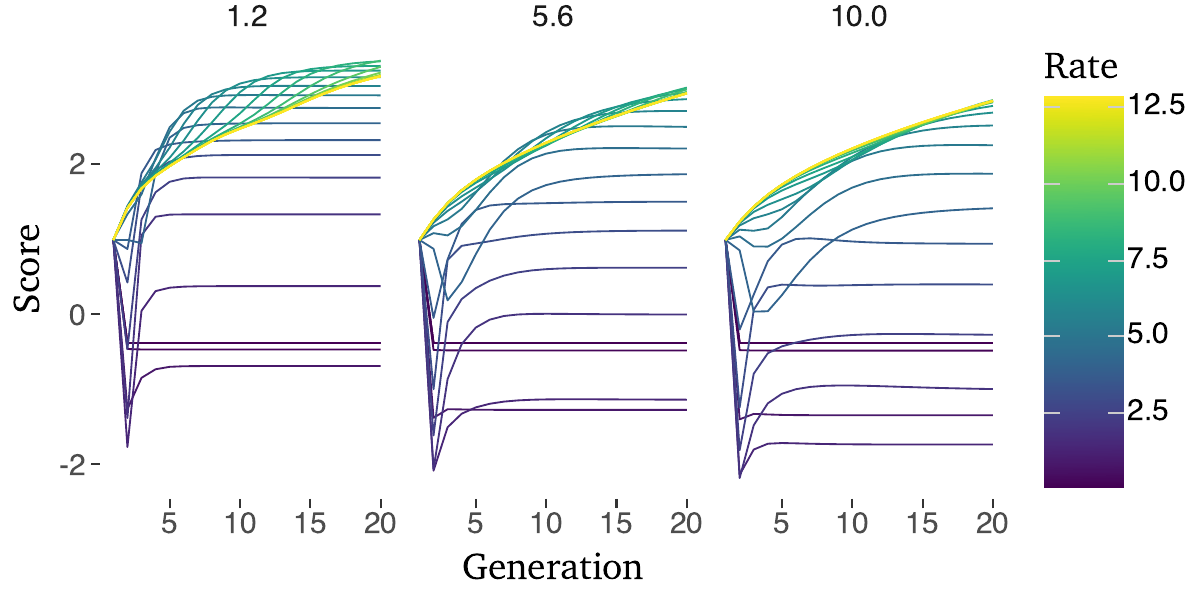}
    \includegraphics[width=0.38\linewidth]{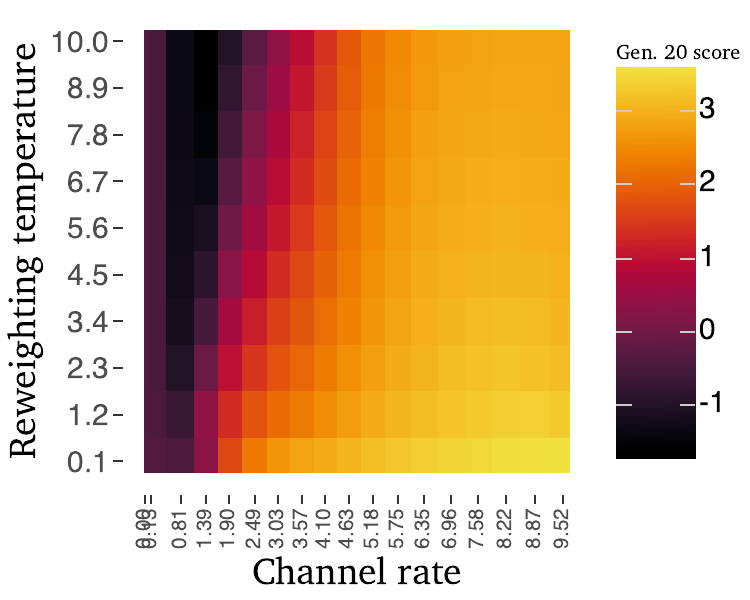}
    \caption{Effect of selective social learning in the rate limited iterated learning model. Left: performance by generation for three different reweighting temperatures. Right: performance after 20 generations for each combination of rate and reweighting temperature.}
    \label{fig:selective-social-learning}
\end{figure}

While our main analyses focus on populations of learners choosing who to learn from randomly, some prior work indicates that selective social learning can be important to the success of iterated learning in large populations \citep{thompson2022complex}. We can implement selective social learning by reweighting the probability of choosing to receive a belief from a particular learner based on the score achieved by that belief. We can model the probability of a learner choosing to learn from an with a particular belief $\boldsymbol{\alpha}$ in the prior generation as the product of the proportion of learners with that belief in the prior generation $p_{T-1}(\boldsymbol{\alpha})$ and the exponentiated score achieved by that belief.
\begin{equation}
    \label{eq:reweighting}
    p_L(\boldsymbol{\alpha}) \propto p_{T-1}(\boldsymbol{\alpha}) \cdot \exp\left(\frac{1}{\tau}\text{score}(\boldsymbol{\alpha})\right)
\end{equation}
The selective social learning temperature $\tau$ controls the strength of selective social learning. A low temperature would lead learners to overwhelmingly choose the best-performing belief, while a high temperature would lead learners to choose a belief based only on its prevalence in the previous generation.

We vary the selective social learning temperature to study how different combinations of temperature and communication rate give rise to different patterns of iterated learning. The results of this analysis are shown in Figure~\ref{fig:selective-social-learning}. The left panel shows performance by generation and rate for three temperatures, while the right panel shows a heatmap of performance after 20 generations for all reweighting temperatures and rates up to 10. 

A rate of at least about 2 is needed to see any iterated learning progress, but low reweighting temperatures (corresponding to strong selective social learning) can make up for a low communication rate to an extent. This is intuitive: the performance signals of members of the previous generation provide information about the true distribution, so less information is needed from the previous learner to estimate the true distribution well. For low reweighting temperatures, intermediate rate limits actually produce faster learning than high rate limits, though they often hit plateaux after a few generations. This can happen because rate-limited channels often map beliefs to more confident beliefs. When those beliefs are confident but wrong, the next generation selects against them. When they are confidently right, the next generation selects for them and benefits from them. This is only helpful when the temperature is low enough to reliably select against against the confidently wrong beliefs, as otherwise they would hinder the overall performance of the population.

\section{Cultural Fano's Inequality proof}
\label{app:fano-proof}

Recall that we model cultural learning over a total of $T$ time periods via probabilistic graphical model over random variables. Suppose observations provided to an agent all come from a set $\mc{X}$ and are drawn from some true distribution in $\Delta(\mc{X})$. Following the standard techniques for minimax risk lower bounds (see Chapter 9.2 of \citep{duchi23ItLectNotes}), we discretize the space of distributions down to $\{P_v\}_{v \in \mc{V}} \subset \Delta(\mc{X})$, each of which is indexed by some $v \in \mc{V}$. We model nature as choosing a data-generating distribution according to $V \sim \text{Uniform}(\mc{V})$. In doing so, we reduce the problem of estimating the underlying data-generating distribution down to a hypothesis testing problem~\citep{yu1997assouad,yang1999information} across $|\mc{V}|$ possible distributions, $\{P_v\}_{v \in \mc{V}}$. Under this so-called canonical hypothesis testing problem, data is then generated i.i.d. from the distribution $P_V \in \Delta(\mc{X})$.

Having reduced the task of cultural learning to identifying the underlying distribution indexed by $V$, it follows that an agent begins with a prior belief over nature's choice $\widetilde{B}_1$\footnote{We use the $\sim$ superscript when writing down prior beliefs for consistency even though it is not compressed and, in keeping with our notation, $B_1 = \widetilde{B}_1$.}. An agent's initial guess as to nature's choice of distribution under prior belief $B_1$ is represented by the random variable $\widehat{V}_1$. Given a first observation $X_1$, the agent with prior $\widetilde{B}_1$ arrives at a true posterior distribution $B_2$. As this individual agent is subject to rate-limited communication of their beliefs to the next learner in the subsequent generation, solving the corresponding lossy-compression problem results in compressed beliefs $\widetilde{B}_2$. Under these rate-limited beliefs, the next agent approaches the learning problem with a guess for nature's choice of distribution represented by $\widehat{V}_2$. Observe that this process may now continue across each time period $t \in [T]$, where the agent begins with prior beliefs given by $\widetilde{B}_t$, observes data drawn from nature's true distribution $X_t$, updates their posterior beliefs to $B_{t+1}$, communicates lossily-compressed beliefs $\widetilde{B}_{t+1}$, and finally synthesizes an updated guess as to nature's choice of distribution $\widehat{V}_{t+1}$.

We prove our main result, Theorem \ref{thm:cultural_fano}, by invoking Fano's inequality~\citep{fano1952TransInfoLectNotes}, which we previously outlined as Fact \ref{fact:fano}. Our use of Fano's inequality allows us to relate the probability of the final agent's guess for nature's choice of distribution $\widehat{V}_{T+1}$ matching reality $V$ with the individual rate-limited communications and beliefs maintained across the constituent $T$ generations. Future work may find it fruitful to extend and explore our analysis using alternative technical approaches or generalizations of Fano's inequality~\citep{aeron2010information,duchi2013distance}.

\setcounter{theorem}{0}

\begin{theorem}
    For iterated learning over $T$ total time periods, we have $$\bP(V \neq \widehat{V}_{T+1}) \geq 1 - \frac{\min\left(\sum\limits_{t=1}^T \bI(\widetilde{B}_t, X_t; \widetilde{B}_{t+1}), \bI(V; \widetilde{B}_1) + \sum\limits_{t=1}^T \bI(V; X_t \mid \widetilde{B}_{t})\right) + \log(2)}{\log(|\mc{V}|)}.$$
\end{theorem}
\begin{proof}
    We represent cultural learning via two separate Markov chains, both of which are conditioned on the prior belief $\widetilde{B}_1$. The first Markov chain $V - \{\widetilde{B}_t\}_{t=2}^{T+1} - \widehat{V}_{T+1}$, given $\widetilde{B}_1$, captures the fact that the only information about the true distribution indexed by $V$ coming through in the final prediction $\widehat{V}_{T+1}$ (aside from what is contained in the prior) is exclusively from the sequence $\{\widetilde{B}_t\}_{t=2}^{T+1}$ of lossy beliefs\footnote{The discrepancy in total time periods comes from having $T$ total generations of learning, each of which has as associated observation $\{X_t\}_{t \in [T]}$, followed by a final posterior belief $\widetilde{B}_{T+1}$ induced by the last observation $X_T$ and $\widetilde{B}_T$.}. Applying the traditional data-processing inequality (Fact \ref{fact:dpi}) to this Markov chain followed by the non-negativity of mutual information and chain rule of mutual information yields the inequality $$\bI(V;\widehat{V}_{T+1} \mid \widetilde{B}_1) \leq \bI(V; \{\widetilde{B}_t\}_{t=2}^{T+1} \mid \widetilde{B}_1) \leq \bI(V; \widetilde{B}_1) + \bI(V; \{\widetilde{B}_t\}_{t=2}^{T+1} \mid \widetilde{B}_1) = \bI(V; \{\widetilde{B}_t\}_{t=1}^{T+1}).$$ Meanwhile, given the prior $\widetilde{B}_1$, there is a second Markov chain $V - \{X_t\}_{t=1}^T - \{\widetilde{B}_t\}_{t=2}^{T+1} - \widehat{V}_{T+1}$ encapsulating how information about $V$ only flows into lossy beliefs through the observations $\{X_t\}_{t=1}^T$; note that $\widetilde{B}_1$ is not included in the chain as a (potentially informative) prior may carry information about $V$ that does not come from any observations $\{X_t\}_{t=1}^T$. Applying the data-processing inequality for a Markov chain with four random variables (Fact \ref{fact:dpi_4rv}) yields $$\bI(V;\widehat{V}_{T+1} \mid \widetilde{B}_1) \leq \bI(\{X_t\}_{t=1}^T; \{\widetilde{B}_t\}_{t=2}^{T+1} \mid \widetilde{B}_1).$$ Naturally, in the presence of two upper bounds for the same quantity, the tightest is obtained by taking a minimum: $$\bI(V;\widehat{V}_{T+1} \mid \widetilde{B}_1) \leq \min\left(\bI(\{X_t\}_{t=1}^T; \{\widetilde{B}_t\}_{t=2}^{T+1} \mid \widetilde{B}_1), \bI(V; \{\widetilde{B}_t\}_{t=1}^{T+1})\right).$$ Thus, applying this inequality after an application of the conditional Fano's inequality (Fact \ref{fact:cond_fano}) immediately yields
    \begin{align*}
        \bP(V \neq \widehat{V}_{T+1}) &\geq 1 - \frac{\bI(V; \widehat{V}_{T+1} \mid \widetilde{B}_1) + \log(2)}{\log(|\mc{V}|)}\\
        &\geq 1 - \frac{\min\left(\bI(\{X_t\}_{t=1}^T; \{\widetilde{B}_t\}_{t=2}^{T+1} \mid \widetilde{B}_1), \bI(V; \{\widetilde{B}_t\}_{t=1}^{T+1})\right) + \log(2)}{\log(|\mc{V}|)}.
    \end{align*}
    
    We focus on the first term in the minimum and look to show the inequality $$\bI(\{X_t\}_{t=1}^T; \{\widetilde{B}_t\}_{t=2}^{T+1} \mid \widetilde{B}_1) \leq \sum\limits_{t=1}^{T} \bI(\widetilde{B}_t, X_t; \widetilde{B}_{t+1}).$$ 

    Starting with the left-hand side of the desired inequality, we can apply the chain rule of mutual information to obtain the sum $$\bI(\{X_t\}_{t=1}^T;\{\widetilde{B}_t\}_{t=2}^{T+1} \mid \widetilde{B}_1) = \sum\limits_{t'=1}^T \bI(\{X_t\}_{t=1}^T; \widetilde{B}_{t'+1} \mid \{\widetilde{B}_{t}\}_{t=1}^{t'}).$$ Using the chain rule on each term in the summation and leveraging the non-negativity of mutual information, we obtain the inequality
    \begin{align*}
        \sum\limits_{t'=1}^T \bI(\{X_t\}_{t=1}^T; \widetilde{B}_{t'+1} \mid \{\widetilde{B}_{t}\}_{t=1}^{t'}) &= \sum\limits_{t'=1}^T \left(\bI(\{X_t\}_{t=1}^T, \{\widetilde{B}_{t}\}_{t=1}^{t'}; \widetilde{B}_{t'+1}) - \ubr{\bI(\{\widetilde{B}_{t}\}_{t=1}^{t'};\widetilde{B}_{t'+1})}_{\geq 0} \right) \\
        &\leq \sum\limits_{t'=1}^T \bI(\{X_t\}_{t=1}^T, \{\widetilde{B}_{t}\}_{t=1}^{t'}; \widetilde{B}_{t'+1}) \\
        &= \sum\limits_{t'=1}^T \bI(\widetilde{B}_{t'}, X_{t'}, \{X_t\}_{t \neq t'}, \{\widetilde{B}_{t}\}_{t=1}^{t'-1}; \widetilde{B}_{t'+1}).
    \end{align*}

    After re-organizing terms to isolate the  $\widetilde{B}_{t'}$ and $X_{t'}$ variables of each mutual information term in the sum, we apply the chain rule once more $$\sum\limits_{t'=1}^T \bI(\widetilde{B}_{t'}, X_{t'}, \{X_t\}_{t \neq t'}, \{\widetilde{B}_{t}\}_{t=1}^{t'-1}; \widetilde{B}_{t'+1}) = \sum\limits_{t'=1}^T \left(\bI(\widetilde{B}_{t'}, X_{t'}; \widetilde{B}_{t'+1}) + \bI(\{X_t\}_{t \neq t'}, \{\widetilde{B}_{t}\}_{t=1}^{t'-1} ; \widetilde{B}_{t'+1} \mid \widetilde{B}_{t'}, X_{t'}) \right).$$ 
    
    At this point, we recall that $\widetilde{B}_{t'+1}$ is a lossy compression of the (compressed) prior belief at the beginning of time period $t'$, $\widetilde{B}_{t'}$, and the observation obtained at time period $t'$, $X_{t'}$, whose underlying probability distribution is given by the channel that achieves the rate-distortion limit at time period $t'$ conditioned on the pair $(\widetilde{B}_{t'}, X_{t'})$. Consequently, there is no more information about $\widetilde{B}_{t'+1}$ accessible beyond what is contained in $(\widetilde{B}_{t'}, X_{t'})$ and we have that $\bI(\{X_t\}_{t \neq t'}, \{\widetilde{B}_{t}\}_{t=1}^{t'-1} ; \widetilde{B}_{t'+1} \mid X_{t'}, \widetilde{B}_{t'}) = 0$. Thus,
    \begin{align*}
        \bI(\{X_t\}_{t=1}^T;\{\widetilde{B}_t\}_{t=2}^{T+1} \mid \widetilde{B}_1) \leq \sum\limits_{t'=1}^T \left(\bI(\widetilde{B}_{t'}, X_{t'}; \widetilde{B}_{t'+1}) + \bI(\{X_t\}_{t \neq t'}, \{\widetilde{B}_{t}\}_{t=1}^{t'-1} ; \widetilde{B}_{t'+1} \mid \widetilde{B}_{t'}, X_{t'}) \right)
        = \sum\limits_{t=1}^{T} \bI(\widetilde{B}_t, X_t; \widetilde{B}_{t+1}).
    \end{align*}
    Putting everything together, we have shown $$\bI(\{X_t\}_{t=1}^T;\{\widetilde{B}_t\}_{t=2}^{T+1} \mid \widetilde{B}_1) \leq \sum\limits_{t=1}^{T} \bI(\widetilde{B}_t, X_t; \widetilde{B}_{t+1}),$$as desired.
    
    All that remains now is to handle the second term of the minimum, $\bI(V; \{\widetilde{B}_t\}_{t=1}^{T+1})$. By the chain rule of mutual information, we have $$\bI(V; \{\widetilde{B}_t\}_{t=1}^{T+1}) = \bI(V; \widetilde{B}_1) + \sum\limits_{t=1}^T \bI(V; \widetilde{B}_{t+1} \mid \{\widetilde{B}_{t'}\}_{t'=1}^t).$$ Recall once again that the compressed posterior obtained at the end of each time period $\widetilde{B}_{t+1}$ is a stochastic function (representing an application of the channel that achieves the rate-distortion limit) which consumes the prior $\widetilde{B}_t$ and the observation $X_t$ as input. If the compression was lossless then, conditioned on $\widetilde{B}_t$, the only information about $V$ remaining in $\widetilde{B}_{t+1}$ that is not already contained in $\widetilde{B}_t$ is exactly the information provided by the observation $X_t$. That is, we have that $\bI(V; \widetilde{B}_{t+1} | \widetilde{B}_t, X_t) = 0$ and, given $\widetilde{B}_t$, there is a Markov chain $V - X_t - \widetilde{B}_{t+1}$. However, for a lossy compression, the residual information contained in $\widetilde{B}_{t+1}$ not accounted for by $\widetilde{B}_t$ cannot exceed the full information offered by $X_t$ and we have an upper bound due to the (conditional) data-processing inequality:
    $$\bI(V; \{\widetilde{B}_t\}_{t=1}^{T+1}) = \bI(V; \widetilde{B}_1) + \sum\limits_{t=1}^T \bI(V; \widetilde{B}_{t+1} \mid \{\widetilde{B}_{t'}\}_{t'=1}^t) \leq \bI(V; \widetilde{B}_1) + \sum\limits_{t=1}^T \bI(V; X_t \mid \{\widetilde{B}_{t'}\}_{t'=1}^t).$$ Now all that remains is to show that, for each $t \in [T]$, $\bI(V; X_t \mid \{\widetilde{B}_{t'}\}_{t'=1}^t) \leq \bI(V; X_t \mid\widetilde{B}_t)$. Since knowledge of nature's choice of distribution $V$ renders all information contained in $\{\widetilde{B}_{t'}\}_{t'=1}^t$ about $X_t$ obsolete, it follows that 
    \begin{align*}
        \bI(V; X_t \mid \{\widetilde{B}_{t'}\}_{t'=1}^t) &= \bH(X_t \mid  \{\widetilde{B}_{t'}\}_{t'=1}^t) - \bH(X_t \mid  V, \{\widetilde{B}_{t'}\}_{t'=1}^t) \\
        &= \bH(X_t \mid \{\widetilde{B}_{t'}\}_{t'=1}^t) - \bH(X_t \mid V) \\
        &= \bH(X_t \mid \{\widetilde{B}_{t'}\}_{t'=1}^t) - \bH(X_t \mid V, \widetilde{B}_t) \\
        &\leq \bH(X_t \mid \widetilde{B}_t) - \bH(X_t \mid V, \widetilde{B}_t) \\
        &= \bI(V; X_t \mid\widetilde{B}_t),
    \end{align*}
    where the inequality follows from the fact that conditioning reduces entropy. Thus, putting everything together we have that $$\bI(V; \{\widetilde{B}_t\}_{t=1}^{T+1}) \leq \bI(V; \widetilde{B}_1) + \sum\limits_{t=1}^T \bI(V; X_t \mid \{\widetilde{B}_{t'}\}_{t'=1}^t) \leq \bI(V; \widetilde{B}_1) + \sum\limits_{t=1}^T \bI(V; X_t \mid \widetilde{B}_t).$$

Substituting back into the earlier result after utilizing Fano's inequality, we have 
\begin{align*}
    \bP(V \neq \widehat{V}_{T+1}) &\geq 1 - \frac{\min\left(\bI(\{X_t\}_{t=1}^T;\{\widetilde{B}_t\}_{t=2}^{T+1} \mid \widetilde{B}_1), \bI(V; \{\widetilde{B}_t\}_{t=1}^{T+1})\right) + \log(2)}{\log(|\mc{V}|)} \\
    &\geq 1 - \frac{\min\left(\sum\limits_{t=1}^{T} \bI(\widetilde{B}_t, X_t; \widetilde{B}_{t+1}), \bI(V; \{\widetilde{B}_t\}_{t=1}^{T+1})\right) + \log(2)}{\log(|\mc{V}|)} \\
    &\geq 1 - \frac{\min\left(  \sum\limits_{t=1}^{T} \bI(\widetilde{B}_t, X_t; \widetilde{B}_{t+1}), \bI(V; \widetilde{B}_1) + \sum\limits_{t=1}^T \bI(V; X_t \mid \widetilde{B}_t)\right) + \log(2)}{\log(|\mc{V}|)},
\end{align*}
as desired.
\end{proof}


\end{document}